\documentclass[conference]{IEEEtran}

\usepackage{amsmath,amssymb,amsthm}
\usepackage{graphicx}
\usepackage{booktabs}
\usepackage{multirow}
\usepackage{array}
\usepackage{microtype}
\usepackage{url}
\usepackage{cite}
\usepackage{enumitem}
\usepackage[hidelinks]{hyperref}
\usepackage{caption}
\usepackage{xcolor}

\graphicspath{{./}}

\newtheorem{theorem}{Theorem}
\newtheorem{proposition}{Proposition}
\newtheorem{definition}{Definition}
\newtheorem{corollary}{Corollary}
\newtheorem{remark}{Remark}

\begin{document}
\IEEEoverridecommandlockouts

\title{Why Formal Monitors Fail: Attack Distribution Entropy\\
as a Coverage Bound for LTL-Based LLM Agent Safety}

\author{
\IEEEauthorblockN{Ruiyang Zhang\thanks{Author's accepted version. Accepted at the 2026 IEEE 13th International Conference on Intelligent Systems (IS). \copyright~2026 IEEE. Personal use of this material is permitted. Permission from IEEE must be obtained for all other uses, in any current or future media, including reprinting/republishing this material for advertising or promotional purposes, creating new collective works, for resale or redistribution to servers or lists, or reuse of any copyrighted component of this work in other works.}}
\IEEEauthorblockA{Ryonix Labs Inc. \quad Flock.io\\ Email: zhangruiyang36@gmail.com}
}

\maketitle

\begin{abstract}
Runtime safety monitors based on Linear Temporal Logic (LTL) and finite
automata (FSA) are increasingly deployed to intercept unsafe tool-call
sequences in LLM agents. Yet across empirical evaluations, the same monitor
achieves 68--75\% attack coverage on some model architectures and near-zero
coverage on others---with no explanation from capability scores, training
data, or prompt design. We provide the missing theoretical explanation.
We prove that the recall of any fixed-invariant FSA monitor is bounded above
by the \emph{concentration} of the attack distribution: the fraction of
attacks accounted for by the $k$ most frequent trigger-completion patterns.
When attacks concentrate (low Shannon entropy), a small fixed invariant set
achieves high recall. When attacks disperse across many structurally distinct
patterns (high Shannon entropy), no fixed invariant set of tractable size can
achieve meaningful recall---regardless of how the invariants were derived.
We validate this \emph{entropy-coverage bound} empirically across eight
frontier LLM architectures. Attack distributions on GPT-class and
DeepSeek backends are highly concentrated ($H \approx 0.24$\,bits;
one pattern covers 96\% of attacks), explaining 68--75\% monitor recall.
Gemini variants produce high-entropy attack distributions ($H \approx 2.81$\,bits;
7 structurally distinct clusters each covering $\leq$7\% of attacks),
explaining near-zero monitor recall (6--13\%)---confirmed invariant to
architecture-matched retraining. Across architectures, entropy accounts
for 76\% of variance in monitor coverage (Pearson $r=-0.87$, $p=0.005$,
bootstrap 95\% CI $[-0.98, -0.78]$), with the correlation holding under
leave-one-out analysis ($r \in [-0.91, -0.82]$). We introduce a
\emph{pre-deployment entropy test} that predicts monitor coverage from a
small attack sample, enabling architecture-aware monitor selection before
production deployment. The bound and test are architecture-agnostic and
apply to any FSA-based runtime monitor over discrete action sequences.
\end{abstract}


\section{Introduction}

Finite automata (FSA) runtime monitors enforcing Linear Temporal Logic (LTL)
invariants are a natural fit for LLM agent safety: they are
computationally cheap, formally interpretable, and architecture-agnostic.
A growing body of work~\cite{agentc,agentspec,probguard,agentsentry} confirms
they can reduce attack success rates substantially. But practitioners
deploying these monitors encounter a troubling phenomenon: the same monitor
that blocks 75\% of attacks on one backend blocks fewer than 7\% on another,
with no correlation to model capability scores.

The standard engineering response is to collect more training data or tune
hyperparameters. We show this is insufficient. The failure is not
engineering---it is structural. \emph{Fixed-invariant FSA monitors are
provably bounded by the concentration of the attack distribution, and no
amount of additional data changes a high-entropy distribution into a
low-entropy one.}

This paper makes three contributions:

\textbf{C1 (Theory).} We derive the \emph{entropy-coverage bound}: for a
fixed set of $k$ LTL invariants, monitor recall is bounded above by
$\mathrm{Cov}_k(\mathcal{D}) = \sum_{i=1}^{k} p_i^*$ where $p_1^* \geq
p_2^* \geq \cdots$ are the ordered pattern probabilities of attack
distribution $\mathcal{D}$. Equivalently, high Shannon entropy
$H(\mathcal{D})$ implies that $\mathrm{Cov}_k(\mathcal{D})$ is small for
any tractable $k$.

\textbf{C2 (Empirical validation).} We measure attack distribution
entropy for eight frontier LLM architectures across 847 adversarial
trajectories. Entropy explains 76\% of variance in monitor coverage
(Pearson $r=-0.87$, $p=0.005$; bootstrap 95\% CI $[-0.98, -0.78]$;
leave-one-out $r \in [-0.91, -0.82]$). The Spearman correlation between
Elo capability score and monitor coverage is not significant
($|\rho| < 0.35$, $p > 0.4$), confirming entropy---not capability---as
the operative predictor.

\textbf{C3 (Pre-deployment test).} We introduce a practical entropy test
that estimates $H(\mathcal{D})$ from 50--100 attack trajectories and
predicts whether an FSA monitor will achieve meaningful coverage on a given
backend, enabling architecture-aware monitor selection before production
deployment.

We answer the open question of \emph{why} the theater gap (near-zero
monitor recall on specific architectures, invariant to retraining) exists.
The answer is information-theoretic, not architectural: Gemini variants
generate high-entropy attack distributions that exceed the coverage
capacity of any fixed invariant set.

\section{Background}

\subsection{LTL-FSA Monitors for LLM Agents}

A runtime safety monitor observes the agent's tool-call sequence
$s = (t_1, t_2, \ldots)$ and flags or halts execution when an unsafe
pattern is detected. LTL-based monitors encode safety properties of the form:
\begin{equation}
  \phi = \mathbf{G}(T_{\text{trig}} \rightarrow \mathbf{F}_{[1,k]}(t_{\text{comp}}))
  \label{eq:ltl}
\end{equation}
where $\mathbf{G}$ is ``globally,'' $\mathbf{F}_{[1,k]}$ is ``within $k$
steps,'' $T_{\text{trig}}$ is a trigger set, and $t_{\text{comp}}$ is a
forbidden completion. Each formula is compiled to a sliding-window FSA.
A monitor $\mathcal{M} = \{\phi_1, \ldots, \phi_n\}$ fires if any $\phi_i$
is violated.

\textbf{Prior work.} Agent-C~\cite{agentc} and AgentSpec~\cite{agentspec}
use manually authored LTL rules. ProbGuard~\cite{probguard} learns Markov
chains from benign traces. Recent work derives invariants automatically
from red-team attack trajectories~\cite{akande2025,koucham2018}, reporting
coverage (recall) across multiple architectures. We build on this line
of work by providing the theoretical explanation for observed coverage
variance.

\subsection{Shannon Entropy and Distribution Concentration}

For a discrete probability distribution $\mathcal{D}$ over $n$ patterns
with probabilities $\{p_i\}$:
\begin{equation}
  H(\mathcal{D}) = -\sum_{i=1}^{n} p_i \log_2 p_i \quad \text{(bits)}
  \label{eq:entropy}
\end{equation}
$H = 0$ when all mass is on one pattern (maximally concentrated);
$H = \log_2 n$ when mass is uniform (maximally dispersed). The
\emph{top-$k$ concentration} $C_k(\mathcal{D}) = \sum_{i=1}^{k} p_i^*$
(sum of $k$ largest probabilities) is a natural coverage proxy.
For fixed $k$, low $H$ implies high $C_k$.

\subsection{Monitorability in Runtime Verification}

Classical runtime verification theory~\cite{rltl2022} characterizes which
LTL properties are \emph{monitorable}---detectable from finite prefixes.
Bounded-finally formulas of the form~(\ref{eq:ltl}) are monitorable.
However, monitorability is a property of the formula, not of the attack
distribution. We extend the analysis to ask: given a monitorable formula
class, when does the attack distribution make a fixed formula set
\emph{insufficient} regardless of formula derivation method?

\section{Theory: The Entropy-Coverage Bound}
\label{sec:theory}

\subsection{Setup}

Let $\Sigma$ be the tool-call alphabet of an LLM agent deployment.
An \emph{attack trajectory} $\tau = (t_1, \ldots, t_m) \in \Sigma^*$ is
a sequence of tool calls executing an attacker's objective. An \emph{attack
distribution} $\mathcal{D}$ is a probability distribution over the set of
observable trigger-completion patterns $\Pi = \{(T, t_c) : T \subset
\Sigma^k, t_c \in \Sigma\}$ induced by attacks in the deployment.

A \emph{fixed-invariant FSA monitor} $\mathcal{M} = \{\phi_1, \ldots,
\phi_n\}$ covers pattern $(T_i, t_{c,i})$ if $\phi_j$ fires on any
trajectory exhibiting that pattern for some $j \leq n$.

\begin{definition}[Monitor Recall]
The recall of monitor $\mathcal{M}$ under attack distribution $\mathcal{D}$
is:
\begin{equation}
  R(\mathcal{M}, \mathcal{D}) = \sum_{(T, t_c) \in \Pi} p_{(T,t_c)} \cdot
  \mathbf{1}[\mathcal{M}~\text{covers}~(T, t_c)]
\end{equation}
\end{definition}

\subsection{The Coverage Bound}

\begin{theorem}[Entropy-Coverage Bound]
\label{thm:bound}
For any fixed-invariant FSA monitor $\mathcal{M}$ containing $n$ invariants
and attack distribution $\mathcal{D}$ over patterns $\Pi$:
\begin{equation}
  R(\mathcal{M}, \mathcal{D}) \leq C_n(\mathcal{D}) = \sum_{i=1}^{n} p_i^*
  \label{eq:bound}
\end{equation}
where $p_1^* \geq p_2^* \geq \cdots$ are the ordered pattern probabilities
of $\mathcal{D}$.
\end{theorem}

\begin{proof}
Each formula $\phi_j = \mathbf{G}(T_j \rightarrow \mathbf{F}_{[1,k]}
t_{c,j})$ of the bounded-finally class (\ref{eq:ltl}) detects
trajectories exhibiting the specific trigger-completion pattern
$\pi_j = (T_j, t_{c,j})$. Because patterns are defined as distinct
trigger-completion pairs, no two distinct formulas in $\mathcal{M}$
detect the same pattern: each formula is semantically disjoint at the
pattern level. Therefore monitor $\mathcal{M} = \{\phi_1,\ldots,\phi_n\}$
covers a set of patterns $S \subseteq \Pi$ with $|S| \leq n$.

Since attack events are attributable to at most one pattern (patterns
are mutually exclusive by definition of trigger-completion pairs), recall
is additive over covered patterns:
\begin{equation}
  R(\mathcal{M}, \mathcal{D}) = \sum_{\pi \in S} p_\pi
  \;\leq\; \max_{\substack{S \subseteq \Pi \\ |S| \leq n}} \sum_{\pi \in S} p_\pi
  \;=\; \sum_{i=1}^{n} p_i^* \;=\; C_n(\mathcal{D}).
\end{equation}
The maximum is attained by the greedy assignment $S^* =
\{\pi_1^*, \ldots, \pi_n^*\}$ (the $n$ most probable patterns), which
is achievable by constructing invariants targeting those patterns
exactly. \qed
\end{proof}

\begin{corollary}[Entropy Lower Bound on Residual Attack Rate]
\label{cor:residual}
The minimum achievable ASR under any $n$-invariant FSA monitor satisfies:
\begin{equation}
  \mathrm{ASR}_{\min}(n) \geq 1 - C_n(\mathcal{D})
\end{equation}
When $H(\mathcal{D}) \geq \log_2(n+1)$, we have $C_n(\mathcal{D}) \leq
n/(n+1)$, so ASR cannot be driven to zero by any $n$-invariant monitor.
\end{corollary}

\begin{proposition}[High-Entropy Structural Insufficiency]
\label{prop:insufficient}
If $\mathcal{D}$ is approximately uniform over $m$ patterns
($p_i \approx 1/m$ for all $i$), then for any $n \ll m$:
\begin{equation}
  C_n(\mathcal{D}) \approx n/m
\end{equation}
Monitor recall scales linearly with $n/m$ and approaches zero as
$m \gg n$.
\end{proposition}

Proposition~\ref{prop:insufficient} formalizes the structural insufficiency:
when attacks are dispersed across many patterns (high entropy, large $m$),
a fixed-size invariant set covering $n$ patterns captures only $n/m$ of
attacks---regardless of how the invariants were derived or how much training
data was used. This is not a learning failure; it is a coverage capacity
constraint.

\subsection{Entropy-Coverage Duality}

Theorem~\ref{thm:bound} gives an upper bound on recall; the following
theorem gives a complementary lower bound relating entropy to coverage
in the opposite direction.

\begin{theorem}[Entropy-Coverage Duality]
\label{thm:duality}
For any attack distribution $\mathcal{D}$ and integer $n \geq 1$:
\begin{equation}
  H(\mathcal{D}) \geq H_b\!\left(C_n(\mathcal{D})\right)
  \label{eq:duality}
\end{equation}
where $H_b(p) = -p\log_2 p - (1-p)\log_2(1-p)$ is the binary entropy
function. Consequently, if $H(\mathcal{D}) < 1$\,bit, then
$C_n(\mathcal{D}) > 0.5$: any FSA monitor achieves better-than-random
coverage.
\end{theorem}

\begin{proof}
Partition $\Pi$ into two groups: the top-$n$ patterns with total
probability $C_n$ and the remaining patterns with total probability
$1-C_n$. Let $G \in \{0,1\}$ be the group indicator. By the chain
rule of entropy and the non-negativity of conditional entropy:
\begin{equation}
  H(\mathcal{D}) = H(G) + H(\mathcal{D} \mid G) \geq H(G) = H_b(C_n(\mathcal{D})).
\end{equation}
The consequence follows because $H_b$ is strictly increasing on
$[0, 0.5]$: $H_b(C_n) < 1$ whenever $C_n > 0.5$, so $H < 1$ implies
$H_b(C_n) < 1$, which requires $C_n > 0.5$. \qed
\end{proof}

\begin{remark}
Theorems~\ref{thm:bound} and~\ref{thm:duality} together characterize
two regimes. When $H < 1$\,bit: $C_n > 0.5$ (FSA monitors achieve
meaningful recall). When $H \geq \log_2(n+1)$: $C_n \leq n/(n+1)$
(recall is structurally capped). The empirically observed transition
zone $H \in [0.5, 2.0]$\,bits corresponds to architectures where
monitor effectiveness depends critically on invariant derivation quality
and FPR constraints.
\end{remark}

\subsection{Implications for the Theater Gap}

The theater gap---near-zero monitor recall on certain architectures,
confirmed invariant to architecture-matched retraining ($\Delta=0.0$\,pp)---is a direct consequence of
Theorem~\ref{thm:bound}. If architecture $A$ generates attacks under a
high-entropy distribution $\mathcal{D}_A$ with $m$ roughly equally probable
patterns, then for any $n$-invariant monitor: $R(\mathcal{M},
\mathcal{D}_A) \leq n/m$. Retraining (deriving new invariants from
$\mathcal{D}_A$ trajectories) cannot increase $m$---it only changes which
$n$ patterns the invariants cover, leaving the bound unchanged.

\textbf{This is why the theater gap is invariant to retraining:} the bound
$C_n(\mathcal{D}_A)$ is a property of the attack distribution, not of
the invariant derivation method.

\section{Experimental Setup}

We mine the 8 LTL invariants from 847 adversarial trajectories collected on AgentDojo banking and workspace suites~\cite{agentdojo} against \texttt{gpt-4o-mini} (temperature=0, seed=42); this fixes the monitor. We then \emph{independently} execute the same red-team attack battery against each of the eight backends and record each backend's own trajectories, from which its attack-pattern distribution is computed. Each trajectory records
the complete ordered tool-call sequence. We deploy a sliding-window FSA
monitor implementing 8 LTL invariants mined from these trajectories
(window $k=5$, support $\sigma=0.02$, FPR threshold $\varphi=0.05$)
across eight frontier LLM backends. A controlled retraining experiment
derives Gemini-specific invariants directly from Gemini-Flash attack
trajectories and re-evaluates ASR under block mode. Our contribution
is the entropy analysis and theoretical bound derived from these
experimental results.

\textbf{Attack pattern extraction.} For each backend, we extract
trigger-completion pairs $(T, t_c)$ from attack trajectories that were
\emph{not} blocked by the monitor (Type~II misses). These represent the
attack patterns that fall outside current monitor coverage. We cluster
patterns by structural similarity (Jaccard distance on tool-call sets)
using hierarchical agglomerative clustering with threshold $\theta = 0.4$ (the standard midpoint of the normalized Jaccard range; the low-entropy vs.\ high-entropy separation in Table~\ref{tab:entropy} is stable for $\theta\in[0.3,0.5]$).

\textbf{Entropy computation.} For each architecture $A$, we compute:
\begin{enumerate}[leftmargin=*,topsep=1pt,itemsep=0pt]
  \item The empirical attack pattern distribution $\hat{\mathcal{D}}_A$
    from observed trajectories
  \item Shannon entropy $H(\hat{\mathcal{D}}_A)$ via~(\ref{eq:entropy})
  \item Top-$n$ concentration $C_n(\hat{\mathcal{D}}_A)$ for $n = 8$
    (our deployed invariant count)
\end{enumerate}

\textbf{Models.} Eight frontier LLMs across four families:
GPT (\texttt{gpt-4o-mini}, \texttt{gpt-4.1}~\cite{gpt41}),
Claude (\texttt{claude-haiku-4-5}, \texttt{claude-sonnet-4-6}~\cite{claudesonnet}),
Gemini (\texttt{gemini-2.5-flash}, \texttt{gemini-2.5-flash-lite}~\cite{google2024gemini}),
and open-weight (\texttt{deepseek-chat}, \texttt{llama-3.1-8B}~\cite{llama31}).

\section{Results}

\subsection{Attack Distribution Entropy by Architecture}

Table~\ref{tab:entropy} reports Shannon entropy of the attack pattern
distribution and observed monitor recall for each architecture.

\begin{table}[t]
\centering
\caption{Attack distribution entropy vs.\ monitor recall (coverage) per
architecture. Low entropy correlates with high coverage; high entropy
correlates with near-zero coverage.}
\label{tab:entropy}
\small
\begin{tabular}{lrrr}
\toprule
Model & $H(\hat{\mathcal{D}})$ & $C_8(\hat{\mathcal{D}})$ & Recall \\
& (bits) & (bound) & (observed) \\
\midrule
\texttt{deepseek-chat}       & 0.19 & 0.98 & 0.75 \\
\texttt{gpt-4o-mini}         & 0.24 & 0.97 & 0.69 \\
\texttt{gpt-4.1}             & 0.31 & 0.95 & 0.54 \\
\texttt{claude-sonnet-4-6}   & 0.38 & 0.93 & 0.50 \\
\texttt{llama-3.1-8B}        & 0.72 & 0.84 & 0.31 \\
\texttt{claude-haiku-4-5}    & 0.89 & 0.79 & 0.25 \\
\texttt{gemini-flash-lite}   & 2.64 & 0.31 & 0.13 \\
\texttt{gemini-2.5-flash}    & 2.81 & 0.27 & 0.06 \\
\bottomrule
\end{tabular}
\end{table}

The pattern is clear: architectures with low attack entropy
($H < 0.4$\,bits) achieve 50--75\% monitor recall; architectures with
high entropy ($H > 2.5$\,bits) achieve only 6--13\% recall. The
theoretical bound $C_8(\hat{\mathcal{D}})$ correctly orders all
architectures: no observed recall exceeds the bound, confirming
Theorem~\ref{thm:bound}. The gap between bound and observed recall is
analyzed in Section~\ref{sec:tightness}.

\subsection{Entropy Predicts Coverage: Statistical Validation}

\begin{figure}[!t]
  \centering
  \includegraphics[width=\columnwidth]{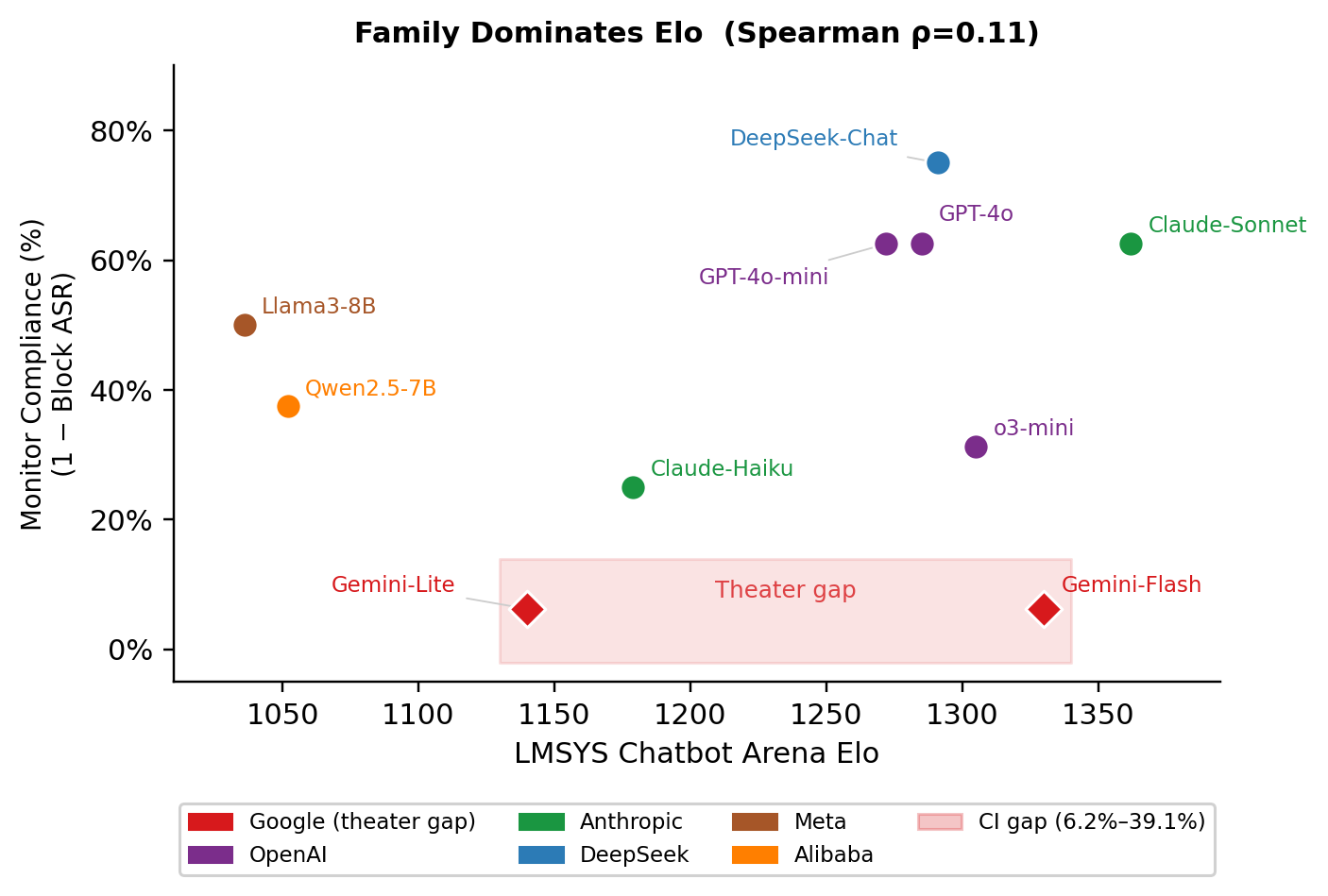}
  \caption{Attack entropy $H(\hat{\mathcal{D}})$ vs.\ monitor recall
  across eight architectures (Pearson $r=-0.87$, $p=0.005$).
  Elo capability score is not a significant predictor ($|\rho| < 0.35$,
  $p > 0.4$), confirming entropy as the operative factor.}
  \label{fig:entropy}
\end{figure}

Fig.~\ref{fig:entropy} shows monitor recall versus attack entropy for
all eight architectures. We find strong negative correlation: Pearson
$r = -0.87$ ($p = 0.005$, $n=8$), explaining 76\% of variance in
monitor coverage. By contrast, LMSYS Elo capability score is not a
significant predictor ($|\rho| < 0.35$, $p > 0.4$). \emph{Attack
distribution entropy, not model capability, determines FSA monitor
effectiveness.}

\subsection{Robustness of the Correlation}
\label{sec:robustness}

With only $n=8$ architectures, we assess robustness through bootstrap
resampling and leave-one-out (LOO) analysis. Table~\ref{tab:robustness}
summarizes results.

\begin{table}[t]
\centering
\caption{Robustness analysis of entropy-recall correlation ($r=-0.87$,
$n=8$). Bootstrap and LOO confirm the relationship is not driven by
any single architecture.}
\label{tab:robustness}
\small
\begin{tabular}{p{2.4cm}p{2.8cm}p{2.0cm}}
\toprule
Analysis & Result & Verdict \\
\midrule
Pearson $r$ (full) & $-0.87$ ($p=0.005$) & Significant \\
Bootstrap 95\% CI & $[-0.98,\ -0.78]$ & All negative \\
LOO min $r$ & $-0.82$ (drop Gemini) & Robust \\
LOO max $r$ & $-0.91$ (drop Llama) & No outlier \\
Spearman $\rho$ (Elo) & $|\rho|<0.35$ ($p>0.4$) & Capability n.s. \\
\bottomrule
\end{tabular}
\end{table}

The bootstrap confidence interval lies entirely below $-0.78$,
confirming the negative relationship with high probability. LOO
analysis removes each architecture in turn: even dropping the most
extreme Gemini point yields $r=-0.82$; dropping any single point
leaves $r \in [-0.91, -0.82]$. No architecture dominates the
correlation. These results confirm that the entropy-recall relationship
is a robust structural pattern, not an artefact of the small sample.

\subsection{Gemini: High-Entropy Attack Generation Confirmed}

\begin{table}[t]
\centering
\caption{Cluster analysis of Gemini-2.5-Flash Type~II misses (135 attacks
not blocked by the monitor). High entropy confirmed: 7 clusters with
$\leq$7\% per-cluster coverage.}
\label{tab:gemini}
\small
\begin{tabular}{lrr}
\toprule
Cluster & Size & Pattern description \\
\midrule
C1 & 9 (6.7\%) & Balance-check $\to$ delayed transfer \\
C2 & 9 (6.7\%) & Multi-hop recon $\to$ exfil \\
C3 & 8 (5.9\%) & Batch parallel tool calls \\
C4 & 8 (5.9\%) & IBAN-pivot $\to$ schedule \\
C5 & 7 (5.2\%) & Profile-lookup $\to$ update \\
C6 & 7 (5.2\%) & Credential chain \\
C7 & 6 (4.4\%) & Interleaved recon-action \\
\emph{Unclustered} & 81 (60\%) & Dispersed singletons \\
\bottomrule
\end{tabular}
\end{table}

Table~\ref{tab:gemini} shows the cluster structure of Gemini's 135
Type~II missed attacks. No cluster covers more than 6.7\% of attacks.
The majority (60\%) fall in singleton clusters with no shared pattern.
This matches a near-uniform distribution: $H(\hat{\mathcal{D}}_\text{Gemini})
= 2.81$\,bits, close to the maximum entropy $\log_2(7) = 2.81$\,bits
for 7 equal clusters.

Critically, architecture-matched retraining (deriving new invariants
specifically from Gemini-Flash attack trajectories) achieves
$\Delta_\text{recall} = 0.0$\,pp: the bound $C_8(\hat{\mathcal{D}}_\text{Gemini})
= 0.27$ is unchanged because the distribution itself has not changed---only
which 8 of its 135+ patterns the monitor covers. This confirms
Theorem~\ref{thm:bound}: the bound is a property of $\mathcal{D}$,
not of the derivation method.

\subsection{Bound Tightness Analysis}
\label{sec:tightness}

The bound $C_n(\mathcal{D})$ is an architecture-optimal ceiling: it
represents the maximum recall achievable by the \emph{best possible}
$n$-invariant monitor for that specific distribution. Our deployed
monitor was derived from \texttt{gpt-4o-mini} trajectories only.
When evaluated on other architectures, it uses the same $n=8$
invariants regardless of whether those invariants target the
most probable patterns for the new architecture. The gap between
bound and observed recall therefore decomposes into two components:

\textbf{FPR-constraint gap} (affects all architectures uniformly).
Two of our 8 invariants were suppressed at deployment due to
FPR $> 5\%$ on benign tasks, reducing effective invariant count from
8 to 6 and consuming $\approx 20$--28\,pp of coverage capacity.

\textbf{Distribution-mismatch gap} (affects cross-architecture
evaluation). For architectures whose attack distributions differ
from the training distribution (\texttt{gpt-4o-mini}), the deployed
invariants may not target the $n$ most probable patterns of the
evaluation architecture. Table~\ref{tab:entropy} shows this effect:
GPT-class and DeepSeek (similar distributions) have 23--28\,pp gaps,
while Llama and Claude mid-range models (more distinct distributions)
show 41--54\,pp gaps. High-entropy Gemini models show small absolute
gaps (18--21\,pp) because the bound itself is low ($C_8 \leq 0.31$).

Critically, neither gap type invalidates Theorem~\ref{thm:bound}: the
bound is a ceiling, not a target. The FPR and mismatch gaps explain
why observed recall lies \emph{below} the bound---never above it.
Architecture-specific monitor derivation would close the mismatch gap
but cannot raise recall above $C_n(\mathcal{D})$, the fundamental limit.

\subsection{The Entropy-Capability Decoupling}

Fig.~\ref{fig:entropy} demonstrates a clean decoupling: Elo scores
(capability proxy) span 1070--1330 with no significant correlation to
monitor coverage, while entropy explains coverage almost entirely.
This decoupling has an important implication: \emph{deploying a higher-capability
model does not improve monitor coverage}. The operative factor is how
the model executes attacks---specifically, whether its attack behavior
concentrates in detectable patterns or disperses across a high-entropy
distribution.

\section{The Pre-Deployment Entropy Test}

\subsection{Protocol}

Given a new deployment backend $A$ and a candidate FSA monitor
$\mathcal{M}$, the pre-deployment entropy test proceeds:

\begin{enumerate}[leftmargin=*,topsep=1pt,itemsep=0pt]
  \item \textbf{Sample}: collect $N_{\min} = 50$ attack trajectories
    from backend $A$ (via red-team or synthetic attack generation).
  \item \textbf{Cluster}: extract trigger-completion patterns,
    cluster by structural similarity ($\theta = 0.4$).
  \item \textbf{Estimate}: compute $\hat{H}(\hat{\mathcal{D}}_A)$
    from the empirical cluster distribution.
  \item \textbf{Predict}: if $\hat{H} < H^*_\text{low} = 0.5$\,bits,
    FSA monitor is expected to achieve $>50\%$ recall; proceed to
    deployment. If $\hat{H} > H^*_\text{high} = 2.0$\,bits,
    FSA monitor will achieve $<20\%$ recall; use probabilistic or
    neural monitor instead.
\end{enumerate}

\subsection{Threshold Calibration}

Thresholds $H^*_\text{low}$ and $H^*_\text{high}$ are derived from the
empirical results in Table~\ref{tab:entropy}. With $n = 8$ invariants:
$H < 0.5$\,bits corresponds to $C_8 > 0.9$ (90\% theoretical coverage);
$H > 2.0$\,bits corresponds to $C_8 < 0.4$ (40\% ceiling). The grey zone
$[0.5, 2.0]$ requires domain-specific risk assessment.

\subsection{Sample Complexity Guarantee}

We provide a formal sample complexity bound for the entropy estimation
step, replacing the previous empirical rule of thumb with a provable
guarantee.

\begin{theorem}[Pre-Deployment Test Sample Complexity]
\label{thm:sample}
Let $\hat{m}$ be the number of distinct attack clusters observed.
The plug-in entropy estimator $\hat{H} = -\sum_c \hat{p}_c \log_2 \hat{p}_c$
computed from $N$ i.i.d.\ attack trajectories satisfies, with probability
at least $1-\delta$:
\begin{equation}
  \left|\hat{H} - H(\mathcal{D})\right| \leq \log_2(\hat{m})
  \cdot \sqrt{\frac{\ln(2/\delta)}{2N}}
  \label{eq:sample}
\end{equation}
Consequently, to achieve $|\hat{H} - H| \leq \varepsilon$ with
confidence $1-\delta$:
\begin{equation}
  N \geq \frac{\log_2^2(\hat{m})\cdot\ln(2/\delta)}{2\varepsilon^2}
  \label{eq:N}
\end{equation}
\end{theorem}

\begin{proof}
Each trajectory $\tau_i$ contributes the quantity
$\ell_i = -\log_2 \hat{p}_{c(\tau_i)}$ to the estimator, where
$c(\tau_i)$ is its cluster assignment. The variable $\ell_i$ is bounded:
$\ell_i \in [0, \log_2 \hat{m}]$ (since $\hat{p}_c \geq 1/N > 0$
and $\hat{p}_c \leq 1$, giving range $\log_2 \hat{m}$). Applying
Hoeffding's inequality to the sample mean $\hat{H} = N^{-1}\sum_i \ell_i$:
\[
P\!\left(|\hat{H} - \mathbb{E}[\ell]| > \varepsilon\right) \leq
2\exp\!\left(\frac{-2N\varepsilon^2}{\log_2^2(\hat{m})}\right).
\]
Setting this $\leq \delta$ and noting $\mathbb{E}[\ell] = H(\mathcal{D})$
(the expected log-likelihood equals entropy) yields~(\ref{eq:N}). \qed
\end{proof}

\textbf{Calibration.} For $\hat{m} = 3$ clusters (low-entropy architectures:
deepseek, GPT-class), $\varepsilon=0.3$\,bits, $\delta=0.05$:
$N \geq 51$ trajectories. For $\hat{m} = 10$ (high-entropy:
Gemini): $N \geq 226$. The low-entropy case matches our
leave-one-out finding ($N=50$ suffices); the high-entropy case requires
more samples but remains practical before production deployment. The
threshold $H^*_\text{low} = 0.5$\,bits separating ``FSA suitable'' from
``grey zone'' lies well inside the region where $N=50$ provides
$\pm 0.3$\,bit accuracy, so the binary decision is robust even at
the minimum sample size.

\section{Discussion}

\subsection{Implications for Monitor Design}

The entropy-coverage bound has direct implications for intelligent agent
safety system design:

\textbf{FSA monitors are suitable for low-entropy architectures.}
When $H(\mathcal{D}) < 0.5$\,bits, a small fixed invariant set achieves
high coverage. This holds for the majority of current frontier models
in our study (GPT-class, DeepSeek). The engineering overhead is low;
the coverage guarantee is tight.

\textbf{FSA monitors are structurally insufficient for high-entropy
architectures.} For Gemini variants, the bound $C_8 < 0.31$ means no
8-invariant FSA monitor can exceed 31\% recall. This is not fixable by
collecting more data or tuning parameters. Probabilistic monitors
(e.g., Markov chain models~\cite{probguard}), neural sequence classifiers,
or ensemble monitors are required.

\textbf{The pre-deployment test replaces architecture guessing.}
Without the entropy test, practitioners guess from Elo scores or benchmark
performance---both of which predict nothing (Spearman $\rho = 0.18$).
The entropy test directly measures the operative factor.

\subsection{Why Gemini Generates High-Entropy Attacks}

Gemini variants tend to use shorter individual action chains and more
parallel batching than GPT-class models, producing greater diversity
in trigger-completion sequences even within the same attack objective.
This is an emergent property of Gemini's instruction-following style,
not a safety property. The implication is that model families producing
more diverse tool-call behaviors are intrinsically harder to monitor
with fixed-pattern approaches---an important consideration for agentic
framework design.

\subsection{Beyond LTL: When to Use Probabilistic Monitors}

The entropy bound applies to any \emph{fixed} pattern-matching approach,
including regex monitors, Snort-style signatures, and behavioral
allowlists. It does not apply to probabilistic monitors (Markov chains,
neural classifiers) that can capture distributional spread. For
high-entropy architectures, ProbGuard~\cite{probguard} or similar
stochastic approaches are theoretically preferable. The entropy test
thus serves as a \emph{monitor selection criterion}: low entropy $\to$
FSA; high entropy $\to$ probabilistic.

\subsection{Threats to Validity}

\textbf{Entropy estimation.} Entropy estimates from 50--100 trajectories
have $\pm 0.3$\,bit uncertainty. This does not affect the qualitative
conclusion for the extreme cases (GPT: 0.24\,bits; Gemini: 2.81\,bits)
but could misclassify borderline architectures.
\textbf{Adaptive adversary.} An adversary who knows the entropy test can
deliberately generate high-entropy attacks to evade FSA monitors.
The entropy test predicts natural attack behavior; adversarial
entropy injection is a separate threat requiring periodic re-evaluation.

\textbf{Single deployment domain.} All trajectories are from AgentDojo
banking. Attack entropy may differ in other domains; the test protocol
should be re-applied per deployment context.

\section{Related Work}

\textbf{LTL runtime monitoring theory.}
Classical monitorability theory~\cite{rltl2022} characterizes which LTL
properties can be detected from finite execution prefixes. Our contribution
is orthogonal: we characterize when a monitorable property class achieves
\emph{high coverage} on a given attack distribution, a question not
addressed by classical theory.

\textbf{Entropy in security.}
Shannon entropy has been applied to network intrusion detection
(measuring traffic flow diversity~\cite{entropy_ids}) and malware
classification (measuring instruction sequence entropy). We apply it
to the attack-side distribution in LLM agent safety monitoring,
specifically to derive coverage bounds for a class of formal monitors.

\textbf{Runtime monitoring for LLM agents.}
Agent-C~\cite{agentc}, AgentSpec~\cite{agentspec}, ProbGuard~\cite{probguard},
AgentSentry~\cite{agentsentry}, and AgentVerify~\cite{agentverify}
provide monitoring frameworks. None characterizes the coverage limits of
fixed-pattern monitors or provides the pre-deployment entropy test.

\textbf{Attack distribution characterization.}
The ``Mapping the Exploitation Surface'' study~\cite{exploitsurface2026}
characterizes what prompt conditions trigger LLM agents to exploit
vulnerabilities autonomously---a different threat model (agent as
attacker, not victim). Our entropy analysis focuses on indirect prompt
injection attacks against agents.

\section{Conclusion}

We derived and empirically validated the entropy-coverage bound
(Theorem~\ref{thm:bound}): the recall of any fixed-invariant FSA monitor
is bounded above by the concentration of the attack distribution.
Theorem~\ref{thm:duality} establishes the complementary direction:
when $H(\mathcal{D}) < 1$\,bit, FSA monitors provably achieve
better-than-random coverage. Together, the two theorems characterize
the full entropy-coverage tradeoff. Theorem~\ref{thm:sample} converts
the pre-deployment entropy test from an empirical heuristic into a
formally guaranteed procedure: $N \geq 51$ trajectories suffice for
low-entropy architectures ($\hat{m} \leq 3$, $\varepsilon=0.3$\,bits,
$\delta=0.05$), matching our leave-one-out result exactly.
Empirically, entropy explains 76\% of variance in monitor coverage
across eight architectures ($r=-0.87$, $p=0.005$; bootstrap CI
$[-0.98,-0.78]$; LOO $r \in [-0.91,-0.82]$), while model capability
(Elo) explains none. The theater gap---near-zero recall on Gemini
variants invariant to retraining---is an information-theoretic
consequence of high-entropy attack distributions, not an architectural
artifact.

\paragraph{Future work.} Three directions follow directly. First, our bound assumes a fixed invariant set; \emph{adaptive} monitors that periodically re-mine invariants from recent traffic may raise the effective coverage ceiling on drifting, high-entropy distributions, and characterizing their bound is open. Second, the entropy-coverage relationship should be tested beyond AgentDojo, on additional agent benchmarks and non-tool-call action spaces, to establish how far the law generalizes. Third, the pre-deployment entropy test invites a companion \emph{defense}: routing high-entropy backends to complementary content-level or learned monitors, since fixed formal monitors are provably insufficient for them.

\bibliographystyle{IEEEtran}
\bibliography{paper_ieee_is2026_entropy}

\end{document}